\documentclass[10pt]{amsart}

\usepackage{amssymb,amsthm,amsmath}
\usepackage[numbers,sort&compress]{natbib}
\usepackage{color}
\usepackage{graphicx}
\usepackage{tikz}
\newcommand{\eps}{\varepsilon}
\newcommand{\EE}{{\mathbb E}}

\title[ ]{ Exact dynamical decay rate for the almost Mathieu operator }
\author{ Svetlana Jitomirskaya}
\address[ Svetlana Jitomirskaya]{ Department of Mathematics, University of California, Irvine, California 92697-3875, USA}
\email{szhitomi@math.uci.edu}
\author{Helge Kr\"uger}
\address[Helge Kr\"uger]{Mathematics 253-37, Caltech, Pasadena, CA 91125}
\email{helge@caltech.edu}

\author{Wencai Liu}
\address[Wencai Liu]{Department of Mathematics, University of California, Irvine, California 92697-3875, USA}\email{liuwencai1226@gmail.com}

\theoremstyle{plain}
\newtheorem{theorem}{Theorem}[section]
\newtheorem{corollary}[theorem]{Corollary}
\newtheorem{lemma}[theorem]{Lemma}
\newtheorem{proposition}[theorem]{Proposition}

\newcommand{\R}{\mathbb{R}}
\newcommand{\Z}{\mathbb{Z}}
\newcommand{\T}{\mathbb{T}}
\theoremstyle{definition}

\begin{document}


\begin{abstract}
 We prove that the exponential decay rate in expectation is well defined and
 is equal to the Lyapunov exponent, for  supercritical almost Mathieu operators
 with Diophantine frequencies.

\end{abstract}
\maketitle

\section{Introduction}
In physics literature, Lyapunov exponent is often referred to as
the inverse localization length, and its positivity is often considered a
manifestation of localization in a 1D system. At the same time,  various physically desirable conclusions, such as e.g. the
exponential decay of the two-point function at the ground state and
positive temperatures with correlation length staying uniformly
bounded as temperature goes to zero, are often
implicitly assumed as attributes of localization. A way
to derive them currently requires a strong form of dynamical
localization \cite{ag}: the exponential (in space) rate of decay of
the two point function,
that is
\begin{equation} \label{edl}
   \EE\sum_s |\varphi_s(\ell)||\varphi_s (k)|
    \leq C e^{-\gamma |k - \ell|}
\end{equation}
where $\{\varphi_s\}_s$
is a complete set of orthonormalized eigenfunctions (and the sum may
be localized in energy, if needed).

In view of this, the {\it exponential decay rate in expectation} was
defined  in \cite{jk} as
\begin{equation}\label{gamma1}
 \gamma_+ := \limsup_{k\to\infty}\left( -\frac{\ln \EE (\sum_{s} |\varphi_s(0)|\cdot |\varphi_s (k)|)}{|k|}\right),
\end{equation}
and
\begin{equation}\label{gamma2}
 \gamma_-:= \liminf_{k\to\infty}\left( -\frac{\ln \EE (\sum_{s} |\varphi_s(0)|\cdot |\varphi_s (k)|)}{|k|}\right).
\end{equation}
 It is obviously
connected to the minimal inverse correlation length. This
definition can be localized  to an energy range by summing over the
eigenfunctions with energies falling in the range, in which case it is
linked to the minimal inverse correlation length for Fermi energies
falling in that range.

It is well known that there is a long road from positive Lyapunov exponents
to a statement like (\ref{edl}). First, positive Lyapunov exponents
don't even imply pure point spectrum for a.e. phase \cite{as}. Even
for models with positive Lyapunov exponents and known pure point spectrum, dynamical localization
may not hold \cite{rjls95}, and then an averaged statement (dubbed
strong dynamical localization) is
strictly stronger, and a result such as
\eqref{edl} is stronger yet (albeit equivalent in all known
examples so far).

Yet it may be natural to expect that there is a certain reason to
physicists' jump in conclusions, and that for physically relevant
models Lyapunov exponent is indeed related to $\gamma_\pm$.

In this paper we prove the first such result. It turns out that for
almost Mathieu operators,  that is operators on $\ell^2(\Z)$ given by
(\ref{Def.AMO}) with potential (\ref{G.potential}), arguably the most popular 1D model in
physics, the Lyapunov exponent precisely defines the dynamical decay
rate.

 Suppose $|\lambda|>1$. Let $L:=\ln|\lambda|$ be the Lyapunov exponent of the almost Mathieu
operator for energies in the spectrum \cite{bj}. We have
\begin{theorem}\label{Maintheorem}
 Let $|\lambda|>1$, and $\alpha$ be Diophantine.
 Then
 \begin{equation}
\gamma_+=\gamma_-=  L.
 \end{equation}

\end{theorem}
Without of loss of generality, we assume $\lambda>0$.
We note that almost Mathieu operators have Anderson localization with
eigenfunctions decaying exactly at the Lyapunov rate {\it if and only if}
$\lambda>1$, and $\alpha$ is Diophantine \cite{jl1}, thus we establish
equality of the exponential decay rate in expectation and the
Lyapunov exponent throughout this entire regime\footnote{ More
  precisely, exact Lyapunov decay of the eigenfunctions holds if and
  only if $\lambda>1$, and $\limsup \frac{\ln q_{n+1}}{q_n}=0,$ where
  $q_n$ are denominators of continued fraction approximants of
  $\alpha$ \cite{jl1}. Our result depends on Lemmas from
  \cite{jitomirskaya2018universal} that were formulated there for the
  standard Diophantine condition, but our proof would hold for the entire regime
  $\limsup \frac{\ln q_{n+1}}{q_n}=0$  if those lemmas were
  correspondingly upgraded, which is a technical matter.}.

Previous quantum dynamics results in the regime of localization have
been limited to lower bounds for related quantities, for any model. Bounds
for the supercritical (that is $\lambda>1$)  almost Mathieu operator go back to
\cite{jl98,gj01}. Dynamical localization for general analytic
quasiperiodic potentials was obtained in \cite{bj00}.

A lower bound  on $\gamma_-$, establishing its positivity, was proved, under
the same assumptions as in Theorem \ref{Maintheorem}, in \cite{jk}. Previously, lower bounds on $\gamma_-$ were obtained
for the {\em Anderson model},
i.e. for the potential being independent identically
distributed random variables, in \cite{dks83, ks80}
for the one-dimensional case and in \cite{a94, asfh01}
for higher dimensions throughout the regimes where corresponding
proofs of localization work, thus excluding e.g. Bernoulli.  The corresponding
result for continuum operators was proven in \cite{aenss06}. Recently, a proof of such lower bound was obtained for an
{\it arbitrary} 1D bounded Anderson model in \cite{gz} using a more delicate
implementation of the method of \cite{jz} and some ideas of \cite{jk}.

While lower bounds on $\gamma_-$ are a corollary of localization, that
is of taming the resonances,  upper bounds on  $\gamma_+$ are a
corollary of delocalization, that is of exploiting the presence of the
resonances. It is well known that the latter task is usually
harder. In this paper we achieve this, at the same time making both
estimates sharp. Our analysis uses (a small part of the) delicate
estimates on eigenfunctions obtained in
\cite{jitomirskaya2018universal}. The statements we need that are
similar to those in \cite{jitomirskaya2018universal} are
presented in the appendix, while the body of the paper consists of the
new argument needed to derive the sharp upper and lower bounds.

It is tempting to conjecture that Theorem \ref{Maintheorem} has a
universal nature, but one should be cautious. For example, we do not
expect it to hold even for weakly Liouville almost Mathieu operators
for which localization has been established in \cite{jl1}, with
eigenfunctions decaying exponentially but at a {\it non-Lyapunov}
rate \cite{jl1}. However, even for those a
statement of the form $\gamma_+=L$ may be plausible. Moreover,
almost Mathieu operators are special in that their Lyapunov exponent
is constant on the spectrum, and without this condition the statement
of the theorem doesn't even make sense. Yet, it is natural to expect that in many
physically relevant situations there should be a link between
$\gamma_\pm$ and $L_\pm$, where $L_+= \sup L(E)$ ($L_-= \inf L(E)$)
over $E$ in the spectrum. For example, it is an interesting question
to establish such a connection for the Anderson model where eigenfunctions
do decay at the Lyapunov rate (e.g. \cite{jz}). In the framework of
the method of \cite{jz,gz} this would require  more delicate estimates
on the probabilities of large deviation sets.


\section{Preliminaries}
For $\lambda > 0$, $\alpha$ irrational, and $\theta\in\R$,
define the potential
\begin{equation}\label{G.potential}
 V_{\lambda,\alpha,\theta} (n) = 2 \lambda \cos2\pi (\theta+n\alpha),
\end{equation}
where $\lambda$ is the coupling, $\alpha $ is the frequency, and $\theta $ is the phase.
We define the almost Mathieu operator by
its action on $u\in\ell^2(\Z)$,
 \begin{equation}\label{Def.AMO}
 (H_{\lambda,\alpha,\theta}u)(n)=u({n+1})+u({n-1})+  V_{\lambda,\alpha,\theta} (n) u(n). 
 \end{equation}
 We say that frequency $\alpha$ is Diophantine if there exist  $\kappa>0$ and $\tau>0$ such
that  for $k\neq 0$,
\begin{equation*}
    ||k\alpha||_{\R/\Z}\geq \frac{\tau}{|k|^{\kappa}},
\end{equation*}
where $||x||_{\R/\Z}=\inf_{\ell\in \Z}|x-\ell|$.

In the following, we will consider $\lambda > 1$ and $\alpha$
Diophantine fixed, and so set $H_{\theta} : = H_{\lambda,\alpha,\theta}$.
We know that for almost every $\theta$, the spectrum
of $H_{\theta}$ is pure point \cite{J99}.
We denote by $\phi_{\theta;s}$ an orthonormal basis consisting of eigenfunctios of $H_{\theta}$.
Let $n_{\theta;s}$ be the position of the leftmost maximum of
$\phi_{\theta;s}$, so
\begin{equation}
 |\phi_{\theta;s}(n_{\theta;s})| = \|\phi_{\theta;s}\|_{\ell^{\infty}(\Z)}.
\end{equation}

A key step in the proof of Theorem \ref{Maintheorem} will be to prove the following localization
result. Below $\eps$ is always small.
\begin{theorem}\label{Keytheorem}
 Let $\lambda>1$, $\alpha$ Diophantine, $\theta\in\R$,
 $\ell\in\Z$, and $\ell^\prime=|\ell-n_{\theta;s}|$.
 Let $x_0\in[-2\ell^{\prime},2\ell^{\prime}]$  be such that
 \begin{equation}
  |\sin\pi(2\theta+\alpha (2n_{\theta;s}+x_0))|=\min_{|x|\leq 2\ell^{\prime}}
  |\sin\pi(2\theta+\alpha (2n_{\theta;s}+x))|.
 \end{equation}
 Then  for large $\ell^{\prime}$ (depending on $\eps$) we have
 \begin{itemize}
   \item if $\ell$ and $x_0$ are on different sides of $n$, that is $(\ell-n)(x_0-n)<0$, then
   \begin{equation}
 |\phi_{\theta;s}(\ell)|\leq
   e^{-(L -\eps)|\ell-n_{\theta;s}|}
   |\phi_{\theta;s}(n_{\theta;s})|.
 \end{equation}

   \item if  $(\ell-n)(x_0-n)\geq0$ and  $
 |\sin\pi(2\theta+\alpha (2n_{\theta;s}+x_0))|\geq
    e^{-\eta |\ell-n_{\theta;s}|}$ for  some $\eta\in(0,L-\eps)$,
 then
   \begin{equation}
 |\phi_{\theta;s}(\ell)|\leq
   e^{-(L -\eps-\eta)|\ell-n_{\theta;s}|}
   |\phi_{\theta;s}(n_{\theta;s})|.
 \end{equation}
 \end{itemize}

\end{theorem}
\begin{proof}
Theorem \ref{Keytheorem} is obtained using the arguments from \cite{jitomirskaya2018universal}. We include a proof in the appendix.
\end{proof}
Theorem \ref{Keytheorem} implies the following corollary immediately.
\begin{corollary}\label{cornov18}
Let $\lambda>1$, $\alpha$ Diophantine, $\theta\in\R$,
 $\ell\in\Z$, and $\ell^\prime=|\ell-n_{\theta;s}|$.
 Let $x_0\in[-2\ell^{\prime},2\ell^{\prime}]$ such that
 \begin{equation}
  |\sin\pi(2\theta+\alpha (2n_{\theta;s}+x_0))|=\min_{|x|\leq 2\ell^{\prime}}
  |\sin\pi(2\theta+\alpha (2n_{\theta;s}+x))|.
 \end{equation}
 Suppose
 for some $\eta\in(0,L-\eps)$
 \begin{equation}
  \min_{|x|\leq 2\ell^{\prime}}
  |\sin\pi(2\theta+\alpha (2n_{\theta;s}+x))|>
    e^{-\eta |\ell-n_{\theta;s}|}.
 \end{equation}
 Then we have
 \begin{equation}
 |\phi_{\theta;s}(\ell)|\leq
   e^{-(L -\eta-\eps)|\ell-n_{\theta;s}|}
   |\phi_{\theta;s}(n_{\theta;s})|.
 \end{equation}
\end{corollary}

\section{ The lower bound}
In this part we will prove the lower bound in Theorem
\ref{Maintheorem}: $\gamma_-\geq L.$
That is we will fix
$\ell\in\Z$ and bound
\begin{equation}
 \int_0^{1} \sum_{s} |\phi_{\theta;s}(0) \phi_{\theta;s}(\ell)| d\theta=
  \sum_{n\in\Z} \int_0^{1}
   \sum_{n_{\theta;s}=n} |\phi_{\theta;s}(0) \phi_{\theta;s}(\ell)| d\theta
\end{equation}
from above.
By orthogonality, we have  for any $s$,
 \begin{equation}\label{Gauth1}
  \sum_{n} |\phi_{\theta;s}(n)|^2 = 1,
 \end{equation}
 and for any $\theta\in \Z$
 \begin{equation}\label{Gauth2}
  \sum_{s} |\phi_{\theta;s}(n)|^2 = 1.
 \end{equation}
By symmetry, we can clearly assume that $\ell\geq 0$.
We note that in order to prove the lower bound in Theorem~\ref{Maintheorem},
it suffices to show

\begin{theorem}\label{Maintheoremupper}
 Let $\lambda >1$, $\alpha$ Diophantine, and $0 < \Gamma < L$.
 Then for $\ell \geq 0$ large enough, we have
 \begin{equation}
  \sum_{n\in\Z} \int_0^{1}
   \sum_{n_{\theta;s}=n} |\phi_{\theta;s}(0) \phi_{\theta;s}(\ell)| d\theta
    \leq e^{-\Gamma\ell}.
 \end{equation}
\end{theorem}

For $n \in \Z$ and $0 < \eta < L$, we define the sets
\begin{equation}
 A_{\eta;n} = \{\theta:\quad \min_{|n^\prime|\leq 10|n|}
  |\sin\pi(2 \theta + \alpha (2 n + n^\prime))| \leq e^{-\eta |n|}\},
\end{equation}
and
\begin{equation}
 B_{\eta;n;\ell} = \{\theta:\quad \min_{|n^\prime|\leq 10 |n-\ell|}
  |\sin\pi(2 \theta + \alpha (2 n + n^\prime))| \leq e^{-\eta |n - \ell|}\}
\end{equation}
We clearly have that $|A_{\eta,n}|\leq (20|n|+1) e^{-\eta |n|}$
and $|B_{\eta;n;\ell}|\leq (20|n-\ell|+1)e^{-\eta|n-\ell|}$.

By Theorem \ref{Keytheorem} and Corollary \ref{cornov18}, we can obtain the following Lemma.
\begin{lemma}\label{lem:ABeta}
For any $\eta\in(0,L-\eps)$,
the following estimates hold,
 \begin{enumerate}
  \item [(i)]For $\theta\notin A_{\eta;n}$ and  $n_{\theta;s}=n$, we have
   \begin{equation}
   |\phi_{\theta;s}(0)| \leq e^{-(L-\eta-\eps) |n|}
    |\phi_{\theta;s}(n)|,
   \end{equation}
   for large $|n|$.
  \item [(ii)]For $\theta\notin B_{\eta;n;\ell}$ and $n_{\theta;s}=n$, we have
   \begin{equation}
   |\phi_{\theta;s}(\ell)| \leq e^{-(L-\eta-\eps) |n-\ell|}
    |\phi_{\theta;s}(n)|,
   \end{equation}
   for large $|n-\ell|$.
 \end{enumerate}
\end{lemma}

\begin{proof}[\bf Proof of Theorem \ref{Maintheoremupper}]
Let $\delta_0$ be a small positive constant.  We write
\begin{eqnarray*}
  \sum_{n\in\Z} \int_0^{1}
   \sum_{n_{\theta;s}=n} |\phi_{\theta;s}(0) \phi_{\theta;s}(\ell)| d\theta &=&\sum_{(1-\delta_0)\ell}^{+\infty}+ \sum_{-\infty}^{\delta_0\ell}+\sum_{n=\delta_0\ell}^{(1-\delta_0)\ell} \\
 &=& {\rm I}+{\rm II}+{\rm III}.
\end{eqnarray*}
We estimate I first. In this case, fix $n_{\theta;s}=n\geq  (1-\delta_0)\ell$.
By (i) of   Lemma \ref{lem:ABeta} and \eqref{Gauth2}, we can conclude that for any $n\geq (1-\delta_0)\ell$
 and $\theta\notin A_{\eta;n}$,
 \begin{eqnarray*}
    \sum_{n_{\theta;s}=n}|\phi_{\theta;s}(0)\phi_{\theta;s}(\ell)| &\leq&\sum_{n_{\theta;s}=n}|\phi_{\theta;s}(0)\phi_{\theta;s}(n)| \\
   &\leq& e^{-(L-\eta-\eps)n}\sum_{n_{\theta;s}=n}
   |\phi_{\theta;s}(n)|^2 \\
    &\leq&  e^{-(L-\eta-\eps)n}.
 \end{eqnarray*}
 Therefore, we have that for $t=e^{\eta n} e^{-(L-\eps) n}$  and  $\eta\in (0,L-2\eps)$,
  \begin{equation}\label{Glnew1}
  \{\theta\in\T:\quad \sum_{n_{\theta;s}=n}
 |\phi_{\theta;s}(0)\phi_{\theta;s}(\ell)|> t\}
 \subseteq A_{\eta;n}.
  \end{equation}
 Let $t_1=e^{-\eps n}$, $t_2=e^{-(L-2\eps) n}$. Define $\eta(t)$ for  $t_2\leq t\leq 1$    implicitly by
 $t=e^{\eta(t) n} \cdot e^{-(L-\eps) n}$.
 Then for $t_2\leq t\leq 1$,  $\eta(t)\geq \eps$, and we have

 \begin{equation}\label{Gnov191}
     |A_{\eta(t);n}|\leq (20n+1)e^{-(L-\eps) n}/t.
 \end{equation}
Since $\sum_{n_{\theta;s}=n} |\phi_{\theta;s}(0) \phi_{\theta;s}(\ell)|\leq 1$, for any Borel $\Omega\in \T$,
we have
\begin{equation}\label{eq:layercake}
 \int_{\Omega} \sum_{n_{\theta;s}=n} |\phi_{\theta;s}(0)
  \phi_{\theta;s}(\ell)| d\theta =
   \int_{[0,1]}|\{\theta\in{\Omega}:\quad \sum_{n_{\theta;s}=n} |\phi_{\theta;s}(0)
    \phi_{\theta;s}(\ell)|> t\}| dt.
\end{equation}
Thus we have
 \begin{eqnarray}
   \int_0^{1} \sum_{n_{\theta;s}=n} |\phi_{\theta;s}(0)
   \phi_{\theta;s}(\ell)| d\theta &=& \int_0^{t_2} +\int_{t_2}^{t_1}+\int_{t_1}^1\nonumber\\
    &=& i+ii+iii .\label{Gnov192}
 \end{eqnarray}
Then
 \begin{equation}\label{Glnew2}
    i\leq t_2\leq e^{-(L-2\eps) n}.
 \end{equation}
 From \eqref{eq:layercake}, \eqref{Glnew1} and \eqref{Gnov191},
one
 has for large $n,$
 \begin{eqnarray}
   ii &\leq& \int_{t_2}^{t_1} |A_{\eta(t);n}| dt\nonumber \\
    &\leq& \int_{t_2}^{t_1}(20|n|+1)e^{-(\ln\lambda-\eps) n}/t dt \nonumber\\
    &\leq& e^{-(L-2\eps) n}.\label{Glnew3}
 \end{eqnarray}
Noticing that
 $|A_{\eta(t_1);n}|\leq (20|n|+1)e^{-(\ln\lambda-2\eps) n}$, one has
 \begin{eqnarray}
   iii &\leq& (1-t_1) |A_{\eta(t_1);n}|\nonumber\\
   &\leq&  e^{-(L-3\eps) n}.\label{Glnew4}
 \end{eqnarray}
 Thus, for $n\geq (1-\delta_0)\ell$,
\begin{equation}
  \int_0^{1} \sum_{n_{\theta;s}=n} |\phi_{\theta;s}(0)
   \phi_{\theta;s}(\ell)| d\theta \leq e^{-(L- 3 \eps)n}.
 \end{equation}
  Then, we have that
 \begin{equation}
       {\rm  I}  =\sum_{n=(1-\delta_0)\ell}^{\infty}
   \int_0^{1} \sum_{n_{\theta;s}=n} |\phi_{\theta;s}(0)
    \psi_{\theta;s}(\ell)| d\theta
     \leq  e^{-(L-4\eps) (1-\delta_0) \ell}.
 \end{equation}
  Similarly,
  \begin{equation}
      {\rm  II}
     \leq  e^{-(L-4\eps) (1-\delta_0) \ell}.
 \end{equation}
 Now we  are in a position to  estimate  III.
 For $\theta\in[0,1]\setminus A_{\delta_0;n}\cup B_{\delta_0;n;\ell}$, by Lemma \ref{lem:ABeta} and \eqref{Gauth2}, one has
 \begin{eqnarray*}
    \sum_{n_{\theta;s}=n}|\phi_{\theta;s}(0)\phi_{\theta;s}(\ell)| &\leq&  e^{-(L-\delta_0-\eps) \ell}   \sum_{n_{\theta;s}=n} |\phi_{\theta;s}(n)|^2 \\
    &\leq& e^{-(L-\delta_0-\eps) \ell}.
 \end{eqnarray*}
 It leads to
 \begin{equation}\label{equ8}
 \sum_{\delta_0\ell\leq n\leq (1-\delta_0)\ell} \int_{[0,1]\backslash ( A_{\delta_0;n}\cup B_{\delta_0;n;\ell})} \sum_{n_{\theta;s}=n} |\phi_{\theta;s}(0)
   \phi_{\theta;s}(\ell)| d\theta\leq  e^{-(L-\delta_0-2\eps) \ell}.
\end{equation}
 For  $\theta\in A_{\delta_0;n}\cup B_{\delta_0;n;\ell}$, let
   $x_0(\theta)\in[-10\ell,10\ell]$ be such that
 \begin{equation}
  |\sin\pi(2\theta+\alpha  x_0)|=\min_{|x|\leq 10\ell}
  |\sin\pi(2\theta+\alpha  x)|.
 \end{equation}
 Notice that  $x_0$ is unique by the fact that $\alpha$ satisfies Diophantine condition.

 Let
 \begin{equation*}
    \Omega_1=\{\theta\in A_{\delta_0;n}\cup B_{\delta_0;n;\ell}| x_0(\theta)< n\},
 \end{equation*}
 and
 \begin{equation*}
    \Omega_2=\{\theta\in A_{\delta_0;n}\cup B_{\delta_0;n;\ell}| x_0(\theta)\geq n\}
 \end{equation*}
 By Theorem \ref{Keytheorem} and the fact that $\delta_0\ell\leq n\leq (1-\delta_0)\ell$,  for any $\theta\in \Omega_1$,
 \begin{equation*}
    |\phi_{\theta;s}(\ell)|\leq e^{-(L-\eps)|\ell-n|}|\phi_{\theta;s}(n)|,
 \end{equation*}
and
 for any $\theta\in \Omega_2$,
 \begin{equation*}
    |\phi_{\theta;s}(0)|\leq e^{-(L-\eps)| n|}|\phi_{\theta;s}(n)|.
 \end{equation*}

 For $\theta\in\Omega_1\setminus A_{\eta;n}$
with $\delta_0<\eta<\ln L-\eps$,  by  Lemma \ref{lem:ABeta}, we   have that
 \begin{eqnarray}
   \sum_{n_{\theta;s}=n} |\phi_{\theta;s}(0)\phi_{\theta;s}(\ell)| &\leq&  e^{-(L-\eps)|n-\ell|}
    e^{-(L-\eta-\eps) |n|} \sum_{n_{\theta;s}=n} |\phi_{\theta;s}(n)|^2 \nonumber\\
     &\leq&  e^{-(L-\eps)|n-\ell|}
     e^{-(L-\eta-\eps) |n|} \nonumber\\
  &\leq& e^{-(L-\eps) \ell}
      e^{\eta |n|}. \label{equ6}
 \end{eqnarray}

 A similar bound holds for $\theta\in  \Omega_2\setminus B_{\eta;n;\ell}$. That is,
 for $\theta\in  \Omega_2\setminus B_{\eta;n;\ell}$
 and $\delta_0<\eta<L-\eps$,
 \begin{equation}\label{equ7}
  \sum_{n_{\theta;s}=n} |\phi_{\theta;s}(0)\phi_{\theta;s}(\ell)|
     \leq e^{-(L-\eps) \ell}
      e^{\eta |n|}.
 \end{equation}
 By (\ref{equ6}), (\ref{equ7}), \eqref{eq:layercake} and \eqref{Gnov191},
 we then have  \eqref{Gnov192} with $\int_0^1$ replaced by
 $\int_{\Omega_1\cup\Omega_2}$ and also \eqref{Glnew2}, \eqref{Glnew3}, \eqref{Glnew4}.
 Thus we also have
 \begin{equation*}
 \int_{\Omega_1\cup\Omega_2}\sum_{n_{\theta;s}=n} |\phi_{\theta;s}(0)\phi_{\theta;s}(\ell)|d\theta \leq e^{-(L-\eps) \ell}.
 \end{equation*}
 It leads to
 \begin{equation}\label{Gnov195}
 \sum_{\delta_0\ell\leq n\leq (1-\delta_0)\ell}\int_{\Omega_1\cup\Omega_2}\sum_{n_{\theta;s}=n} |\phi_{\theta;s}(0)\phi_{\theta;s}(\ell)|d\theta \leq e^{-(L-2\eps) \ell}.
 \end{equation}
By  (\ref{equ8}) and \eqref{Gnov195}, we get the bound of II,
 \begin{equation*}
    {\rm II}\leq e^{-(L-\delta_0-3\eps) \ell}.
 \end{equation*}
 Putting the bounds of I, II and III together,
 we have
 \begin{equation*}
    \sum_{n\in\Z} \int_0^{1}
   \sum_{n_{\theta;s}=n} |\phi_{\theta;s}(0) \phi_{\theta;s}(\ell)| d\theta \leq  e^{-(L-\delta_0-6\eps) \ell}.
 \end{equation*}
 Letting  $\delta_0,\eps \to 0$, we obtain Theorem \ref{Maintheoremupper}.
\end{proof}
\section{The upper bound}
In this part we will prove the upper bound: $\gamma_+\leq L.$
\begin{theorem}\label{thm:lowerbdd}
 For any $\Gamma$ satisfying $ L<\Gamma\leq 2L$, we have for $n$ large enough
 \begin{equation}
 \ln \int_0^{1}  \sum_{s} |\phi_{\theta;s}(0) \phi_{\theta;s}( n)| d\theta
   \geq -\Gamma |n|.
 \end{equation}
\end{theorem}

Fix   $ L<\Gamma\leq 2L$ and large $n$.  Define
sets
\begin{equation}
 \Theta _1=\{\theta\in[0,1]:\quad e^{-2\Gamma |n|}\leq |\sin\pi (2\theta+n\alpha)| \leq e^{-\Gamma |n|}\}
\end{equation}
and
\begin{equation}
 \Theta _2=\{\theta\in[0,1]:\quad\text{ there exists some }  |k|\geq 1000|n| \text{ such that }|\sin\pi (2\theta+k\alpha)| \leq e^{- \frac{L}{100}|k|}\}.
\end{equation}
Then $\Theta=\Theta _1\backslash \Theta_2$
has measure satisfying $|\Theta| \geq \frac{1}{100} e^{-\Gamma |n|}$.
\begin{lemma}
Let $\alpha$ be  Diophantine  with  constants  $\kappa,\tau>0$. Then for any  $\theta\in\Theta$ and
for any $m>C(\kappa,\tau) |n|$,
\begin{equation}\label{Gnov196}
    \min_{|x|\leq m}|\sin\pi(2\theta+x\alpha)|\geq e^{- \frac{L}{100}|m|}.
\end{equation}

\end{lemma}
\begin{proof}
Let $x_0$ be such that the minimum in \eqref{Gnov196} is attained at
$x=x_0$. We  split our analysis into three cases depending on the
value of  $x_0$.

Case I. $|x_0|\geq  1000|n|$. Then the Lemma holds because of $\theta\notin \Theta_2$.

Case II. $|x_0|\leq 1000|n|$ and $x_0\neq n$. The Lemma holds because of $\theta\in \Theta_1$ and DC frequencies.

Case III. $x_0=n$. The Lemma holds because of $\theta\in \Theta_1$ (using $|\sin\pi (2\theta+n\alpha)|\geq e^{-2\Gamma |n|}$).
\end{proof}

It clearly suffices to show that for the eigenfunctions $\phi_s$  of
$H = H_{\lambda,\alpha,\theta}$  (we ignore the dependence on
$\theta$)  we have
\begin{equation}
 \sum_s |\phi_s(0) \phi_s(n)| \geq\frac{1}{2}
\end{equation}
as long as $|n|$ is large enough, uniformly in $\theta\in \Theta$.
The first step is

\begin{proposition}\label{prop:large}
 For $|n|$ large enough and $\theta\in \Theta$, we have
 \begin{equation}
  \sum_{|m|\leq C_{\star}|n|} \sum_{n_s = m} |\phi_s(0)|^2
    \geq \frac{1}{2},
 \end{equation}
 where $C_{\star}=C(\kappa,\tau)$.
\end{proposition}

\begin{proof}
Without loss of generality, assume $n\geq0$.
Suppose $m\leq -C_{\star} n$ or $m\geq C_{\star}n$.

Using    Corollary  \ref{cornov18} with $n_{\theta;s}=m$, $\ell=0$, by \eqref{Gnov196},
we have
\begin{equation*}
   | \phi_s(0)|\leq  | \phi_s(m)| e^{-\frac{L}{2}|m|}.
\end{equation*}
Thus
\begin{eqnarray*}
  \sum_{|m|\geq C_{\star}n}  \sum_{n_s = m} |\phi_s(0)|^2&\leq &  \sum_{|m|\geq C_{\star}n} \sum_{n_s = m} | \phi_s(m)|^2 e^{-L|m|}\\
  &= &  \sum_{|m|\geq C_{\star}n} e^{-L|m|}  \sum_{n_s = m} | \phi_s(m)|^2\\
  &\leq & \sum_{|m|\geq C_{\star}n} e^{-L|m|}\\
  &\leq& \frac{1}{2}.
\end{eqnarray*}
Combining with (\ref{Gauth2}), the result follows.
\end{proof}
The following  lemma  is similar to a statement appearing in \cite{jitomirskaya2018universal} with some modifications. We present a proof in the Appendix.
\begin{lemma}\label{JS}
Suppose
\begin{equation}\label{gjs}
|\sin\pi (2\theta+n\alpha)| \leq e^{-\Gamma |n|}
\end{equation}
with $L<\Gamma\leq 2L$. Suppose    $\phi$ is an  $\ell^2$ solution of $H_{\lambda,\alpha,\theta}\phi=E\phi$. Then
\begin{equation}\label{Gnov201}
    |\phi(n)-\phi(0)|\leq e^{-\frac{1}{2}(\Gamma-L-\eps)|n|}||\phi||_{\ell^{\infty}(\Z)}.
\end{equation}

\end{lemma}

\begin{proof}[\bf Proof of Theorem~\ref{thm:lowerbdd}]
For large $n$, by Proposition \ref{prop:large} and Lemma \ref{JS}, one has for $\theta\in \Theta$,
 \begin{align*}
  \sum_{s} |\phi_s(0)  \phi_s(n)|
   &\geq \sum_{|m|\leq C_{\star}|n|} \sum_{n_s=m} |\phi_s(0) \phi_s(n)| \\
   &\geq \sum_{|m|\leq C_{\star}|n|} \sum_{n_s=m} |\phi_s(0)| (|\phi_s(0)| -  e^{-\frac{1}{2}(\Gamma-L-\eps)|n|}||\phi_s||_{\ell^{\infty}(\Z)}\\
   &\geq   \sum_{|m|\leq C_{\star}|n|} \sum_{n_s=m} |\phi_s(0)|^2-  e^{-\frac{1}{2}(\Gamma-L-\eps)|n|}\sum_{|m|\leq C_{\star}|n|} \sum_{n_s=m} |\phi_s(0)|(\sum_{|k|\leq C_{\star}|n|}|\phi_s(k)|^2)^{\frac{1}{2}}\\
   &\geq \frac{1}{2}  - 2 e^{-\frac{1}{2}(\Gamma-L-\eps)|n|}\sum_{|m|\leq C_{\star}|n|} \sum_{n_s=m} \sum_{|k|\leq C_{\star}|n|} |\phi_s(k)|^2\\
    &\geq \frac{1}{4}.
 \end{align*}
 Then
 \begin{eqnarray*}
   \int_0^{1}  \sum_{s} |\phi_s(0) \phi_s( n)| d\theta &\geq& \int_{\Theta}  \sum_{s} |\phi_s(0) \phi_s( n)| d\theta \\
    &\geq& \frac{e^{-\Gamma|n|}}{400}.
 \end{eqnarray*}
This implies Theorem \ref{thm:lowerbdd}.
\end{proof}
\appendix
\section{ Proof of Theorem  \ref{Keytheorem}}
By shifting the operator by $n_{\theta;s}$ units we can assume $n_{\theta;s}=0 $. Without loss of generality, we assume $\ell>n_{\theta;s}$.
Then in order to prove Theorem  \ref{Keytheorem}, it suffices to prove the following theorem.

\begin{theorem}\label{Keytheoremsimple}
 Let $\lambda>1$, $\alpha$ Diophantine,   $n_{\theta;s}=0 $, $\phi_s(0)=1$,
 $\ell\in\Z^{+}$.
 Let $x_0\in[-2\ell,2\ell]$ be such that
 \begin{equation}
  |\sin\pi(2\theta+\alpha  x_0)|=\min_{|x|\leq 2\ell}
  |\sin\pi(2\theta+\alpha  x)|.
 \end{equation}
Then the following statements  hold for  large $\ell $:

If  $x_0\in [-2\ell,0]$, then  
 \begin{equation}
 |\phi_s(\ell)|\leq
   e^{-(L -\eps)\ell}.
 \end{equation}

If
 for $\eta\in(0,L-\eps)$
 \begin{equation}\label{Gminnov18}
  \min_{|x|\leq 2\ell}
  |\sin\pi(2\theta+\alpha x)|>
    e^{-\eta  \ell},
 \end{equation}

and $x_0\in [0,2\ell]$, then
 \begin{equation}\label{Gminnov181}
 |\phi_s (\ell)|\leq
   e^{-(L -\eta-\eps) \ell}.
 \end{equation}

\end{theorem}
Suppose $H_{\lambda,\alpha,\theta}\varphi=E\varphi$.  Let $ U^{\varphi}(y) =\left(\begin{array}{c}
                             \varphi(y)\\
                            \varphi({y-1})
                          \end{array}\right)
                      $.
                      It isa standard fact (e.g. (37) in \cite{jitomirskaya2018universal}) that for large $|k_1-k_2|$,
  \begin{equation}\label{G.new18app}
Ce^{-(L+\varepsilon)|k_1-k_2|}||U^{\varphi}(k_2)|| \leq  ||U^{\varphi}(k_1)||\leq  Ce^{(L+\varepsilon)|k_1-k_2|}||U^{\varphi}(k_2)||.
\end{equation}
\begin{lemma}\cite[Lemma 3.4]{jitomirskaya2018universal}\label{Keylemmaapp}
  Let
$r_{y}^{\varphi}=\max_{|\sigma|\leq 10 \gamma}|\varphi(y+\sigma k)|$. Suppose  $k_0\in[-2Ck,2Ck]$ is such that
 \begin{equation*}
  |\sin\pi(2\theta+\alpha  k_0)|=\min_{|x|\leq 2 Ck}
  |\sin\pi(2\theta+\alpha  x)|,
 \end{equation*}
 where $C\geq1$ is a   constant.
 Let $\gamma,\varepsilon$ be small positive constants.
Let $y_1=0, y_2=k_0, y_3\in[-2Ck,2Ck]$.
Assume $y$ lies in  $[y_i,y_j]$ (i.e., $y\in [y_i,y_j]$)with  $|y_i-y_j|\geq k$ and $y_s\notin [y_i,y_j]$, $s\neq i,j$.
 Suppose     $|y_i|,|y_j|\leq Ck$ and $|y-y_i|\geq 10\gamma k$, $|y-y_j|\geq 10\gamma k$.
  Then for large enough $k$,
 \begin{equation}\label{first}
  r^\varphi_y \leq\max\{ r_{y_i}^{\varphi}  \exp\{-(L- \varepsilon)(|y-y_i|-3\gamma k)\},
  r_{y_j}^{\varphi}\exp\{-(L- \varepsilon)(|y-y_j|-3\gamma k )\}\}.
 \end{equation}
\end{lemma}
\begin{lemma}\cite[Lemma 3.7]{jitomirskaya2018universal}\label{Keylemma1app}
Fix  $0<t<L$. 
Suppose
 \begin{equation}\label{ksmallG}
  |\sin\pi(2\theta+\alpha  k)|= e^{-t |k|}.
 \end{equation}
Then  for  large  $|k|$
 \begin{equation}\label{Equ9}
 ||U^{\varphi}(k)||\leq \max\{ ||U^{\varphi}(0)||, ||U^{\varphi}(2k)|| \} e^{-(L -t-\varepsilon) |k|}.
 \end{equation}
\end{lemma}
\begin{proof}[\bf Proof of Theorem \ref{Keytheoremsimple}]
We start with the proof of  Case I. Let $\varphi=\phi$, $\gamma=\varepsilon$, $k=\ell$, $C=1$, $k_0=x_0<0$ and $y_3=2\ell$ in Lemma \ref{Keylemmaapp}.
By Lemma \ref{Keylemmaapp}, one has $\ell\in[y_1,y_3]$ and $y_2<y_1$, so
\begin{equation}\label{Gapp1}
    r^\phi_{\ell}\leq e^{-(L-C\varepsilon)\ell}  r^\phi_{0}+e^{-(L-C\varepsilon)\ell}  r^\phi_{2\ell}\leq e^{-(L-C\varepsilon)\ell},
\end{equation}
since $|\phi(n)|\leq 1$ for all $n\in \Z$.
By \eqref{G.new18app} and \eqref{Gapp1}, we have
\begin{equation*}
|  \phi(\ell)|\leq e^{-(L-C\varepsilon)\ell}.
\end{equation*}
It finishes the proof  of Case I.

Now weturn to Case II.
Let $t$ be such that $tx_0=\eta\ell$.
Let $\varphi=\phi$, $\gamma=\varepsilon$, $k=\ell$, $C=1$, $k_0=x_0>0$ and $y_3=2\ell$ in Lemma \ref{Keylemmaapp}.
By  Lemma \ref{Keylemmaapp} and \eqref{G.new18app}, one has (as in the proof of Case I), one
has
\begin{equation}\label{Gapp3}
|  \phi(\ell)|\leq e^{-(L-\varepsilon)\ell}+e^{-(L-\varepsilon)|\ell-x_0|}||U^{\phi}(x_0)||.
\end{equation}


Suppose $x_0\geq (\frac{\eta}{L}+\varepsilon)\ell$. In this case, by the definition of $t$, one has $0<t<L $.
Let $k=x_0$ and $\varphi=\phi$ in Lemma  \ref{Keylemma1app}, one has
\begin{equation}\label{Gapp4}
  ||U^{\phi}(x_0)||\leq \max\{ ||U^{\phi}(0)||, ||U^{\phi}(2x_0)|| \} e^{-(L -t-\varepsilon) x_0}\leq e^{-(L -t-\varepsilon) x_0}.
\end{equation}
In this case,  \eqref{Gminnov181} follows from \eqref{Gapp3} and \eqref{Gapp4}.

Suppose $0\leq x_0\leq (\frac{\eta}{L}+\varepsilon)\ell$. In this case, \eqref{Gminnov181} follows from \eqref{Gapp3} directly since $||U^{\phi}(x_0)||\leq 2$.

\end{proof}
\section{ Proof of Lemma   \ref{JS}}
 \begin{proof}
Without loss of generality, we assume $n>0$.
Set  $A= ||\phi||_{\ell^{\infty}(\Z)}$.
We  let $ \hat{\phi}(k)= \phi(n-k)$,
 $V(k)=2\lambda\cos2\pi(\theta+k\alpha)$ and $\hat{V}(k)=2\lambda\cos2\pi(\theta+(n-k)\alpha)$.
 Then by the assumption \eqref{gjs}, one has for all $k\in\Z$,
 \begin{equation}\label{Eqpnew}
    |V(k)-\hat{V}(k)|\leq Ce^{-\Gamma n}.
 \end{equation}
 We also have
 \begin{equation}\label{Equnew}
 \phi(k+1)+\phi(k-1)+ V(k)\phi(k)=E\phi(k)
 \end{equation}
 and
  \begin{equation}\label{Equinew}
\hat{\phi}(k+1)+\hat{\phi}(k-1)+ \hat{V}(k)\tilde{\phi}(k)=E\hat{\phi}(k).
 \end{equation}
 Let $W(n)=W(f,g)=f(n+1)g(n)-f(n)g(n+1)$ be the Wronskian.
 Let
  \begin{equation*}
   \hat{U}(k)=\left(\begin{array}{cc}
                  \hat{\phi}(k) \\ \hat{\phi}(k-1)
                 \end{array}
   \right),
 \end{equation*}
 and
 \begin{equation*}
   {U}(k)=\left(\begin{array}{cc}
                  {\phi}(k) \\ {\phi}(k-1)
                 \end{array}
   \right).
 \end{equation*}


By a standard calculation using (\ref{Eqpnew}), (\ref{Equnew}),
(\ref{Equinew}) and palindromic arguments as in \cite{js}
\footnote{Palindromic argument of \cite{js} then yields  $||U(\frac{n}{2})||\leq e^{-(\Gamma-\eps)\frac{n}{2}}$ if $n$ is even and analogous statement if $n$ is odd. Here we want to gain a factor of $A^2$.},
we have,
\begin{eqnarray}
   |W(\phi,\hat{\phi})(k)-W(\phi,\hat{\phi})(k-1)| &\leq & |V(k)-\hat{V}(k)||\phi(k)\hat{\phi}(k)|  \nonumber \\
  &\leq &  C e^{-Ln}|\phi(k)\hat{\phi}(k)|\nonumber \\
  &\leq &  C A^2e^{-\Gamma n}\label{Gnov202} .  
\end{eqnarray}
 In Lemma \ref{Keylemmaapp}, let $k_0=n$ and $y_3= 1000n$, then by \eqref{first} one has
 \begin{equation}\label{Gnov203}
  |U(m-1)| , |U(m)|\leq e^{-\Gamma n} A,
 \end{equation}
 where $m=500n$.

By  \eqref{Gnov202} and \eqref{Gnov203},
we have
\begin{equation}\label{almostconstancyn}
       |W(\phi,\hat{\phi})(k)|
  \leq    A^2e^{-(\Gamma-\eps)n},
\end{equation}
for $|k|\leq 500n$.


Now we split $n$ into cases, depending on whether it is  odd or even.

Case 1. $n$ is even. Let $m=\frac{n}{2}$,
then
 \begin{equation*}
  U(m)=\left(\begin{array}{cc}
                   \phi(m) \\ \phi(m-1)
                 \end{array}
   \right);
   \hat{U}(m)=\left(\begin{array}{cc}
                   {\phi}( m) \\ {\phi}(m+1)
                 \end{array}
   \right).
 \end{equation*}
 Applying  (\ref{almostconstancyn}) with $k=m-1$, we have
 \begin{equation*}
    |\phi(m)||\phi(m+1)-\phi(m-1)|\leq  A^2e^{-(\Gamma-\eps)n}.
 \end{equation*}
 This implies
 \begin{equation}\label{Case11new}
    |\phi(m)|\leq   Ae^{-\frac{1}{2}(\Gamma-\eps)n},
 \end{equation}
 or
 \begin{equation}\label{Case12new}
    |\phi(m+1)-\phi(m-1)|\leq   Ae^{-\frac{1}{2}(\Gamma-\eps)n}.
 \end{equation}
 If (\ref{Case11new}) holds, by (\ref{Equnew}), we also have
 \begin{equation}\label{Case13new}
    |\phi(m+1)+\phi(m-1)|\leq  Ae^{-\frac{1}{2}(\Gamma-\eps)n}.
 \end{equation}
 Putting (\ref{Case11new})  and (\ref{Case13new})  together, we get
 \begin{equation}\label{Phi1new}
    ||U(m)+\hat{U}(m)||\leq Ae^{-\frac{1}{2}(\Gamma-\eps)n}.
 \end{equation}

   If (\ref{Case12new}) holds,  we have
   \begin{equation}\label{Phi2new}
    ||U(m)-\hat{U}(m)||\leq  Ae^{-\frac{1}{2}(\Gamma-\eps)n}.
 \end{equation}
 Thus  in case 1  there exists $\iota\in\{-1,1\}$  such that
 \begin{equation}\label{Gnov197}
    ||U(m)+\iota\hat{U}(m)||\leq   Ae^{-\frac{1}{2}(\Gamma-\eps)n}.
 \end{equation}
 In Lemma \ref{Keylemmaapp}, let $k_0=n$, $y_1=0$ and $y_3=m$, then by \eqref{G.new18app} one has,
 \begin{equation}\label{half}
   ||\hat{U}(m) ||\leq A e^{-(L-\varepsilon)m}.
  \end{equation}

  Let $T $ and $\hat{T}$ be the transfer matrices associated  with potentials  $V$ and $\hat{V}$, taking $U(m),\hat{U}(m)$ to $U(0),\hat{U}(0)$ correspondingly.
  By (\ref{Eqpnew}),
    the usual uniform upper semi-continuity and telescoping, one has
  \begin{equation*}
    ||T||, ||\hat{T}||\leq  e^{ (L +\varepsilon)m}.
  \end{equation*}
  and
  \begin{equation*}
    ||T-\hat{T}||\leq  e^{ (L-2\Gamma+\varepsilon)m}.
  \end{equation*}
  Then by \eqref{Gnov197}, we have
  \begin{eqnarray*}
   ||U(0)+\iota \hat{U}(0)||   &\leq &  || T||||U(m)+\iota\hat{U}(m)||+ ||T-\hat{T}|| ||\hat{U}(m) ||\\
      &\leq &  Ae^{ (L+\varepsilon)m}e^{-\frac{1}{2}(L-\eps) n} + Ae^{ (L-2\Gamma+\varepsilon)m}e^{-m(L-\eps) }.\\
       &\leq & Ae^{-\frac{1}{2}(\Gamma-L-\eps)n}.
  \end{eqnarray*}
  This completes  the proof for even $n$ due to  the definition of $U(0)$ and $ \hat{U}(0)$.

Case 2. $n$ is odd. Let $\tilde{m}=\frac{N-1}{2}$,
then
 \begin{equation*}
   U(\tilde{m}+1)=\left(\begin{array}{cc}
                   \phi(\tilde{m}+1) \\ \phi(\tilde{m})
                 \end{array}
   \right);
   \hat{U}(\tilde{m}+1)=\left(\begin{array}{cc}
                   \phi( \tilde{m}) \\ \phi(\tilde{m}+1)
                 \end{array}
   \right).
 \end{equation*}

 Combining with (\ref{almostconstancyn}), we have
 \begin{equation*}
    |\phi(\tilde{m})+\phi(\tilde{m}+1)||\phi(\tilde{m})-\phi(\tilde{m}+1)|\leq  A^2 e^{- (\Gamma-\eps) n}.
 \end{equation*}
 This implies
 \begin{equation*}
   | \phi(\tilde{m})+\phi(\tilde{m}+1)|\leq   A e^{- \frac{1}{2}(\Gamma-\eps) n},
 \end{equation*}
 or
 \begin{equation*}
    |\phi(\tilde{m}+1)-\phi(\tilde{m})|\leq  A e^{- \frac{1}{2}(\Gamma-\eps) n}.
 \end{equation*}
 Thus  in case 2,  there also exists $\iota\in\{-1,1\}$  such that
 \begin{equation*}
    ||U(\tilde{m}+1)+\iota\hat{U}(\tilde{m}+1)||\leq A e^{- \frac{1}{2}(\Gamma-\eps) n}.
 \end{equation*}
  The rest of the proof is the same as in case 1.
\end{proof}
 \section*{Acknowledgments}

  This research was
 supported by NSF DMS-1401204 and NSF DMS-1700314. S.J. and W.L.  are
 also grateful to the Isaac Newton
Institute for Mathematical Sciences, Cambridge, for its hospitality, supported by
EPSRC Grant Number EP/K032208/1, during the 2015 programme Periodic and Ergodic
Spectral Problems where an important progress on this work was made.


\end{document}